\documentclass{article}
\usepackage{ijcai21}

\usepackage{times}
\usepackage{soul}
\usepackage{url}
\usepackage[hidelinks]{hyperref}
\usepackage[utf8]{inputenc}
\usepackage[small]{caption}
\usepackage{graphicx}
\usepackage{amsmath}
\usepackage{amsthm}
\usepackage{booktabs}
\usepackage{algorithm}
\usepackage{algorithmic}
\newcommand{\citet}[1]{\citeauthor{#1}~\shortcite{#1}}
\newcommand{\citep}{\cite}

\usepackage{amssymb,cleveref,tikz,multirow}
\usepackage[shortlabels]{enumitem}

\newtheorem{theorem}{Theorem}[section]
\newtheorem{proposition}[theorem]{Proposition}
\newtheorem{lemma}[theorem]{Lemma}

\theoremstyle{definition}
\newtheorem{definition}[theorem]{Definition}

\newtheorem{open}[theorem]{Open Problem}

\newtheorem{claim}{Claim}

\newcommand{\distg}{\textsc{Dist}^G}
\newcommand{\pathg}{\textsc{Path}^G}
\newcommand{\med}{\textsc{Nearest}}
\newcommand{\mms}{\text{MMS}}
\newcommand{\mmsrank}{\textsc{MmsRank}}
\newcommand{\mmsrankadd}{\textsc{MmsRankAdd}}
\newcommand{\fvsnum}{\textsc{FvsNum}}

\allowdisplaybreaks

\title{Graphical Cake Cutting via Maximin Share}

\author{
Edith Elkind$^1$\and
Erel Segal-Halevi$^2$\And
Warut Suksompong$^3$\\
\affiliations
$^1$Department of Computer Science, University of Oxford\\
$^2$Department of Computer Science, Ariel University\\
$^3$School of Computing, National University of Singapore\\
}

\begin{document}

\maketitle

\begin{abstract}
We study the recently introduced cake-cutting setting in which the cake is represented by an undirected graph. This generalizes the canonical interval cake and allows for modeling the division of road networks. We show that when the graph is a forest, an allocation satisfying the well-known criterion of maximin share fairness always exists. Our result holds even when separation constraints are imposed, in which case no multiplicative approximation of proportionality can be guaranteed. Furthermore, while maximin share fairness is not always achievable for general graphs, we prove that ordinal relaxations can be attained.

\end{abstract}

\section{Introduction}

Cake cutting is an old and famous problem in resource allocation, with the cake serving as a metaphor for a heterogeneous divisible resource that is supposed to be fairly divided among interested agents.
While the problem has long enjoyed substantial attention from mathematicians and economists, it has also attracted ongoing interest from computer scientists, not least those working in artificial intelligence \citep{BalkanskiBrKu14,LiZhZh15,BranzeiCaKu16,AlijaniFaGh17,MenonLa17,BeiHuSu18,GoldbergHoSu20,HosseiniIgSe20}.
Indeed, as \citet{Procaccia13} aptly put it, cake cutting is more than just child's play.

The cake in cake cutting is typically assumed to be a one-dimensional interval.
Even though the linear representation is appropriate for modeling the division of time (for instance, usage of a jointly-owned facility), or space in a hallway, it is too simplistic to capture more complex resources such as road networks.
This consideration has led \citet{BeiSu21} to introduce a more general \emph{graphical cake} model, in which the resource comes in the form of a connected undirected graph. 
In parallel, \citet{Segalhalevi21} addressed the case of graphs given by a disjoint union of intervals.
In contrast to the single interval cake, for general graphs it is not always possible to find a connected \emph{proportional} allocation, that is, an allocation that gives every agent a connected subset of the cake worth at least $1/n$ of the agent's value for the entire cake, where $n$ denotes the number of agents among whom the cake is divided.
Nevertheless, Bei and Suksompong showed that more than half of this guarantee can be recovered: for any connected graph, it is possible to ensure every agent at least $1/(2n-1)$ of her total value, and this factor is tight in general (but can be improved for certain graphs).
For the union-of-intervals case, Segal-Halevi established an approximation factor of $1/(m+n-1)$, where $m$ denotes the number of intervals.

In this paper, we study graphical cake cutting with respect to another prominent fairness notion, \emph{maximin share fairness} \citep{Budish11}.
An allocation satisfies this notion if it assigns to every agent a bundle worth at least her maximin share, i.e., the largest value that she can get if she is allowed to partition the cake into $n$ connected parts and receives the least valuable part.
Maximin share fairness is a robust notion which can naturally take into account the features and constraints arising in various settings.
Indeed, this robustness will make it feasible to derive positive results in three settings where approximate proportionality fails.
First, we allow the graph to be arbitrary---it may be disconnected, and each of its connected components may have an arbitrary topology.
Second, we allow agents to have general monotone valuations---unlike most of the cake cutting literature, we do not assume that valuations must be additive.
With disconnected graphs or non-additive valuations, 
providing an approximate proportionality guarantee solely in terms of $n$ is impossible.
Third, we also consider a recently introduced setting of \citet{ElkindSeSu21} in which the shares of different agents should be sufficiently separated from one another; this allows us to model space between pieces of roads with different owners, for example transition or buffer zones.
In the presence of separation constraints, obtaining any multiplicative approximation of proportionality is again infeasible, as can be seen when the value of all agents is entirely concentrated in the same tiny portion. 
This observation motivated Elkind et al.~to use the maximin share benchmark when separation is imposed.
They showed that with separation, a maximin allocation always exists for a path graph, but not for a cycle graph.

\subsection{Our Contributions}

We begin in Section~\ref{sec:no-separation} by addressing the basic case where no separation constraints are imposed.
As \citet{BeiSu21} noted, in this case it can be crucial for proportionality approximations whether or not pieces of different agents are allowed to share a finite number of points.
We show that this assumption is also crucial with respect to maximin share fairness.
In particular, if points can be shared, then a maximin allocation may not exist even when the graph is a star, whereas if sharing is not allowed, the existence of such an allocation can be guaranteed for acyclic graphs (i.e., a disjoint union of trees, also known as a forest).
Our results complement those of Bei and Suksompong, who observed that approximate proportionality cannot be provided even for trees under the no-sharing assumption.
In addition, our guarantees degrade gracefully for graphical cakes with cycles: we attain a $1$-out-of-$(n+r)$ maximin allocation, where $r$ is the \emph{feedback vertex set number} of the cake (that is, the smallest number of vertices whose removal would make the cake acyclic).

In Section~\ref{sec:positive-separation}, we consider the more general case where the pieces of any two agents must be separated by distance at least a given (positive) parameter.
Our main technical result shows that a maximin allocation exists whenever the graph is acyclic---this significantly generalizes the existence result of \citet{ElkindSeSu21} for paths and complements their non-existence result for cycles.
As with paths, our proof uses the following high-level idea: Given the maximin partitions of the agents, we find a part in one agent's partition such that allocating the part to that agent rules out at most one part in each remaining agent's partition; this allows us to recurse on the remaining agents and cake.
While in the case of paths the desired part can be found by simply scanning the path from left to right,
in an arbitrary forest there is no `left' or `right', so new techniques are needed.
We develop auxiliary lemmas related to \emph{real trees}---metric spaces defined by tree graphs---which may be of independent interest.
As in the case of no separation, we obtain a $1$-out-of-$(n+r)$ maximin allocation for general graphs.

For the case of positive separation, we show that in general, the factor $n+r$ cannot be improved: for every $r\geq 0$, when $n$ is sufficiently large, there exists a graph with feedback vertex set number $r$ and a set of $n$ agents that do not admit a $1$-out-of-$(n+r-1)$ maximin allocation.
However, better guarantees can be attained for smaller $n$ and specific classes of graphs.

\subsection{Additional Related Work}

As mentioned earlier, cake cutting is a popular topic among researchers of several disciplines---see the classic books of \citet{BramsTa96} and \citet{RobertsonWe98}, as well as a more recent survey by  \citet{Procaccia16} offering a computer scientist's perspective.

Most of the work relevant to graphical fair division and separation constraints has been covered by \citet{BeiSu21} and \citet{ElkindSeSu21}; we refer to the related work section of their papers, but highlight here some important aspects of our study.
First, we assume that each agent must receive a connected piece of cake.
This assumption is often made in order to ensure that agents do not end up with a collection of crumbs---indeed, a bundle made up of tiny stretches of road in different parts of the network is unlikely to be of much use. 
Second, the connectivity requirement is imposed not only on the allocation, but also in the definition of the maximin share.
This is consistent with previous work on maximin share fairness in constrained settings \citep{BouveretCeEl17,BiswasBa18,LoncTr20,BeiIgLu21,ElkindSeSu21}.
\citet{BouveretCeEl17} proved that for \emph{indivisible} items lying on a tree, a maximin allocation exists.
However, as we discuss in Section~\ref{sec:conclusion}, the ``last diminisher'' approach that they used for this proof does not work in our setting with separation.
Recently, \citet{IgarashiZw21} studied envy-freeness in graphical cake cutting under the assumption that agents cannot share individual points, while \citet{ElkindSeSu21-Land} investigated land division with separation constraints.

In the papers above, as in our paper, the resource to be divided lies on a graph.
A complementary line of work studies fair division scenarios in which the \emph{agents} lie on a graph indicating their acquaintance \citep{abebe2016fair,BeiQiZh17,AzizBoCa18}.

\section{Preliminaries}

There is a set of agents $\mathcal{N}=[n]$, where $[k]:=\{1,2,\dots,k\}$ for any positive integer $k$.
The cake is represented by a finite undirected graph $G=(V,E)$, which may be connected or not.
Each agent has a nonnegative, monotone, and continuous valuation function $v_i$, which is not necessarily additive.
In particular, continuity implies that the vertices in $V$ have zero value.
Note that the cases studied by \citet{ElkindSeSu21} and in most cake cutting papers are special cases of this model: an interval cake corresponds to taking $G$ to be any path graph, while a pie cake is equivalent to a cycle graph.
A \emph{piece of cake} is a finite union of intervals from one of more edges in $E$.
The piece is said to be \emph{connected} if for any points $x,y$ in it, one can get from $x$ to $y$ along the graph $G$ by only traversing this piece of cake.
We assume that each agent must receive a connected piece of cake.

There is a separation parameter $s\ge 0$.
When $s > 0$, the edge lengths play an important role.
We measure distance along the edges of $G$. For any two points $x,y\in G$, we denote by $\distg(x,y)$ the length of a shortest path from $x$ to $y$ along the edges of $G$; if $x$ and $y$ belong to different connected components of $G$, we set $\distg(x,y) = \infty$.
For two pieces of cake $X,Y\subseteq G$, we denote by $\distg(X,Y)$ the shortest distance between a point in $X$ and a point in $Y$ along the edges of $G$, i.e., $\distg(X,Y) = \inf_{x\in X, y\in Y}\distg(x,y)$; if $Y$ consists of a single point $y$, we simply write $\distg(X,y)$.

A {\em partition} of the cake is a set $\mathbf{P} = \{P_1,\dots,P_n\}$, where each $P_i$ is a connected piece of cake, and the pieces are pairwise disjoint: $P_i \cap P_j = \emptyset$ for all $i\ne j$. When $s=0$, we will
consider, in addition to the disjoint-pieces setting, an alternative setting in which $P_i \cap P_j$ may contain finitely many points.
An {\em allocation} is defined similarly, except that we have a vector $\mathbf{A} = (A_1,\dots,A_n)$ instead of a set, where piece $A_i$ is allocated to agent~$i$.
A partition $\mathbf{P}$ is said to be \emph{$s$-separated} if $\distg(P_i,P_j) \geq s$ for all $i\ne j$; an analogous definition holds for an allocation.
We assume that partitions and allocations are required to be $s$-separated.
Observe that for $s > 0$, in any $s$-separated partition or allocation, some of the cake necessarily remains unallocated.
Moreover, since any two pieces are separated by a positive distance, we assume without loss of generality in this case that the pieces contain only closed intervals.
Denote by $\Gamma_{n,s}$ the set consisting of all $s$-separated partitions.

The main fairness notion of our paper is the following:

\begin{definition}
\label{def:MMS}
The \emph{maximin share} of agent $i$, denoted by $\mms^{n,s}_i$, is defined as $\sup_{\mathbf{P}\in \Gamma_{n, s}}\min_{j\in [n]}v_i(P_j)$.
\end{definition}

We omit the superscript $s$ when it is clear from the context.
Similarly to the interval cake \citep{ElkindSeSu21}, the supremum in Definition~\ref{def:MMS} can be replaced with a maximum.
In other words, there is always a maximizing partition, which we refer to as a \emph{maximin partition}.
An allocation in which every agent receives at least her maximin share is said to be a \emph{maximin allocation}.

\section{No Separation}
\label{sec:no-separation}
In this section, we address the basic case where there is no separation constraint imposed on the allocation, i.e., $s=0$.
When $s>0$, the pieces of any two agents cannot be adjacent to each other, so we can assume without loss of generality that all pieces consist only of closed intervals and the pieces have empty intersections.
For $s=0$, however, this is not true: 
there are essential differences between 
the empty-intersection setting
and the \emph{finite-intersection} setting, in which pieces may overlap in finitely many points.
This observation was made by \citet{BeiSu21} with respect to approximate proportionality for the case of a star graph.
Specifically, in the empty-intersection setting,
$n-1$ agents do not receive the center of the star and therefore can receive cake from at most one edge, leading to strong negative results.
On the other hand, 
in the finite-intersection setting, 
decent welfare guarantees can be obtained.

As we will demonstrate, the distinction between empty intersection and finite intersection is crucial with respect to maximin share fairness too.
Note that the maximin share is calculated using the same restrictions that are imposed on allocations.
First, we show that in the finite-intersection setting,
a maximin allocation may not exist.

\begin{proposition}
\label{prop:graph-mss-negative}
Assume that the allocated pieces are allowed to intersect in a finite number of points.
There exists an instance with $n=3$ agents and a star cake in which no maximin allocation exists.
\end{proposition}

\begin{proof}
The proof follows the celebrated Theorem 2.1 of \citet{KurokawaPrWa18}, which shows that a maximin allocation of indivisible objects may not exist for $n=3$ agents.

In their instance, there are $12$ objects, indexed by $j\in[3]$ and $k\in[4]$. Each agent $i\in[3]$ values each object $(j,k)$ by:
\begin{align*}
v_i(j,k) = 1000000 + 1000\cdot T_{j,k} + E_{j,k}^{(i)}
\end{align*}
where $T, E^{(1)}, E^{(2)}, E^{(3)}$ are carefully chosen $3\times 4$ matrices with all values smaller than $100$. 
Kurokawa et al.~proved that every agent can partition the objects into $3$ subsets of $4$ objects each, in such a way that the sum of values in each subset is exactly $4055000$; this value is therefore the maximin share of all agents.
These authors then showed that no allocation gives every agent at least this value.

In our instance, there is a star graph with $12$ edges connected to a single center vertex $c$. The edges are indexed by $j\in[3]$ and $k\in[4]$. Each agent $i\in[3]$ has value $v_i(j,k)$ for the edge $(j,k)$, and this value is spread uniformly across the edge. 
Since the pieces may intersect in a finite number of points, the maximin share of each agent is also $4055000$ in our instance, by partitioning the set of edges to $3$ subsets of $4$ edges each (these subsets intersect in the single point $c$).

If an agent's piece is contained in a single edge, then her value is clearly less than $2000000$. 
Hence, in a maximin allocation, each edge must belong to only one agent.
We may therefore assume without loss of generality that each agent receives a piece containing two or more whole edges. But the same argument as that of \citet{KurokawaPrWa18} shows that no such allocation can be a maximin allocation.
\end{proof}

Next, we show that in the 
empty-intersection setting, a maximin allocation always exists when the cake is a forest.

\begin{theorem}
\label{thm:forest-0}
Let $G$ be a forest and $s=0$.
Assume that all allocated pieces must be completely disjoint.
For agents with arbitrary monotone valuations, a maximin allocation exists.
\end{theorem}

Intuitively, given the $n$ maximin partitions of the agents, we want to choose a part in one agent's partition that overlaps at most one part in each remaining agent's partition---this will allow us to recurse on the remaining agents and their leftover partitions.
To this end, we introduce the following definition.
Given a graph $G$ and a family $\mathbf{X} := (X_1,\ldots,X_k)$ of connected pieces of $G$, a piece $X_{j^*}$ is called \emph{$0$-good} provided that
for all $j_1,j_2\in[k]$, the following holds: If $X_{j_1}\cap X_{j^*} \neq \emptyset$ and  $X_{j_2}\cap X_{j^*} \neq \emptyset$, then
$X_{j_1}\cap X_{j_2} \neq \emptyset$.

\begin{lemma}
\label{lem:tree-0}
Let $G$ be a tree and  $\mathbf{X} := (X_1,\ldots,X_k)$ a family of connected subsets of $G$, for some integer $k\geq 1$.
For some $j^*\in[k]$, the piece $X_{j^*}$ is $0$-good.
\end{lemma}

In order to prove this lemma, we must handle both open and closed pieces.\footnote{
We are grateful to Alex Ravsky for the proof idea.
} 
Indeed, for $n=3$ and a star graph with three edges of equal value, a maximin partition contains two open pieces and one closed piece.\footnote{
The difference between the empty-intersection and the finite-intersection settings
stems from the fact that 
the finite-intersection analogue of Lemma \ref{lem:tree-0} does not hold.
 For example, when $G$ is a star graph with center $c$ and edges $e_1,e_2,e_3,e_4$, and  $\mathbf{X} = (e_1\cup c \cup e_2,~ e_3\cup c \cup e_4,~ e_1\cup c \cup e_3,~ e_2\cup c \cup e_4)$, 
there is no $X_{j*}$ with the property that
if $X_{j_1}\cap X_{j*}$ is infinite and  $X_{j_2}\cap X_{j*}$ is infinite, then
$X_{j_1}\cap X_{j_2}$ is infinite.
}

\begin{proof}[Proof of Lemma~\ref{lem:tree-0}]
Let $m$ be the number of edges in $G$. 
The proof is by induction on $m$.
For the base case $m = 1$, $G$ is an interval; assume without loss of generality that it is the interval $[0,1]$.
Each $X_j$ is also an interval which may be open, half-open, or closed. 
For each $j\in[k]$, let $\ell_j$ and $r_j$ be the left and right endpoint of $X_j$, respectively.
Choose $j^*\in[k]$ such that $r_{j^*}$ is smallest. 
If $X_{j^*}$ contains its right endpoint (i.e., it is right-closed), 
then any piece $X_{j}$ intersecting $X_{j^*}$ must have $\ell_{j}\leq r_{j^*}$ and must contain the point $r_{j^*}$.
On the other hand, if $X_{j^*}$ does not contain its right endpoint (i.e., it is right-open),
then any piece $X_{j}$ intersecting $X_{j^*}$ must have $\ell_{j} < r_{j^*}$ and must contain an interval $(r_{j^*}-\varepsilon, r_{j^*})$ for some sufficiently small $\varepsilon>0$.
In both cases, any two such pieces $X_j$ intersect, so $X_{j^*}$ is $0$-good.
This concludes the base case.

We proceed to the inductive step.
Let $m\ge 2$, assume that the statement holds for graphs with at most $m-1$ edges, and suppose that $G$ has $m$ edges. 
Since $G$ is a tree, it has a leaf, i.e., a vertex $w$ connected to a single edge $e$. Let $G^-$ be the graph $G$ without the vertex $w$ and the edge $e$.
Let $u$ be the vertex at the other end of $e$; note that $u$ is a vertex of $G^-$.
We consider two cases.

\underline{Case 1:}
At least one piece of $\mathbf{X}$ is contained in $e$ (so it is an interval).
Among all such intervals, choose an $X_{j^*}$ whose closest point to $u$ is as far away as possible.
Then $X_{j^*}$ is $0$-good by the same arguments as in the base case $m=1$.

\underline{Case 2:}
No piece of $\mathbf{X}$ is contained in $e$.
This means that all pieces of $\mathbf{X}$ intersect $G^-$. 
Let $\mathbf{X}' := (X'_1,\ldots,X'_k)$, where
$X'_j := X_j \cap G^-$ for each $j\in[k]$.
By the inductive assumption applied to $G^-$, at least one piece in $\mathbf{X'}$, say $X'_{j^*}$, is $0$-good with respect to $\mathbf{X'}$.
It suffices to show that $X_{j^*}$ is also $0$-good with respect to $\mathbf{X}$.

We claim that if $X_{j^*}$ intersects some other piece $X_j$, then $X_{j^*}'$ also intersects $X_{j}'$.
To see this, note that the intersection between $X_{j^*}$ and $X_j$ may occur either in $G^-$ or in $e$ (or both). If the intersection occurs in $G^-$, then $X_{j^*}'$ intersects $X_{j}'$ and we are done.
If the intersection occurs in $e$, then $X_{j^*}\cap e$ intersects $X_{j}\cap e$, so both of the intersections are non-empty.
But by the assumption of Case~2, all pieces of $\mathbf{X}$ intersect $G^-$.
Therefore, both $X_{j^*}$ and $X_{j}$ contain the point $u$,  which is the unique point connecting $e$ and $G^-$. Therefore $u\in X_{j^*}' \cap X_{j}'$, so again $X_{j^*}'$ intersects $X_{j}'$.
This establishes the claim.

We now show that $X_{j^*}$ is $0$-good with respect to $\mathbf{X}$.
Suppose that $X_{j^*}$ intersects two other pieces in $\mathbf{X}$, say $X_{j_1}$ and $X_{j_2}$.
By the claim in the previous paragraph, 
$X'_{j^*}$ intersects both $X'_{j_1}$ and $X'_{j_2}$.
Since $X_{j^*}'$ is $0$-good with respect to $\mathbf{X}'$, it must be that
$X_{j_1}'$ intersects $X_{j_2}'$.
In particular, $X_{j_1}$ intersects $X_{j_2}$.
Hence, $X_{j^*}$ is $0$-good with respect to $\mathbf{X}$.
\end{proof}

With Lemma~\ref{lem:tree-0} in hand, we can now show that, in the empty-intersection setting, a maximin allocation exists whenever the cake is a forest.

\begin{proof}[Proof of Theorem~\ref{thm:forest-0}]
For each agent, consider her maximin partition.
Every part of the partition is contained in some tree of the forest. 
Let $T\subseteq G$ be a tree that contains at least one part from the maximin partition of at least one agent. 

For every agent $i\in \mathcal{N}$, let $k_i$ be the number of parts of $i$'s maximin partition that are contained in $T$, and denote the parts by $T_{i,1},\ldots,T_{i,k_i}$.
By Lemma \ref{lem:tree-0}, there exists some $i\in \mathcal{N}$ and $j\in[k_i]$ such that $T_{i,j}$ is $0$-good.
Allocate the part $T_{i,j}$ to agent $i$, 
and divide the remaining cake recursively among the remaining agents.

The remaining cake is still a forest.
By definition of a $0$-good subset, for every other agent, at most one part of her maximin partition overlaps the allocated piece $T_{i,j}$. 
Hence, for each of the $n-1$ remaining agents, at least $n-1$ parts from her maximin partition remain intact. Therefore the recursive call indeed returns a maximin allocation.
\end{proof}

As we have seen, the seemingly minor distinction of whether individual points can be shared among allocated pieces makes a decisive difference in relation to maximin share fairness.
Which assumption is more realistic depends on the use case, for example whether road intersections can only be owned by one agent or shared by multiple agents.
\citet{BeiSu21} showed that nontrivial egalitarian welfare can be obtained only when sharing is allowed.
Thus, our results complement theirs by exhibiting that even when sharing is infeasible, a reasonable fairness guarantee can still be made in terms of the maximin share.

We now proceed to general graphs.
We consider an ordinal relaxation called \emph{$1$-out-of-$k$ maximin share}, denoted by $\mms_i^{k,s}$, or simply $\mms_i^{k}$ when $s$ is clear from the context.
The idea is that instead of taking partitions into $n$ parts as in the canonical maximin share, we allow partitions into $k$ parts, where $k > n$ is a given parameter.
For each graph $G$, 
let $\fvsnum(g)$ be the \emph{feedback vertex set number} of $G$, that is, the minimum number of vertices whose removal makes the graph acyclic.\footnote{
Computing $\fvsnum(g)$ is NP-hard \citep{karp1972reducibility},
but here we use it only for existence proofs.
Note that $\fvsnum(G)$ is upper-bounded by the \emph{circuit rank} of $G$, that is, the minimum number of \emph{edges} whose removal makes the graph acyclic.
The circuit rank of a graph $G = (V,E)$ with $c$ connected components is $|E|-|V|+c$.
}

\begin{theorem}
\label{thm:graph-0}
Let $s = 0$,
and assume that all allocated pieces must be completely disjoint.
For any graph $G$ and any $n$ agents with arbitrary monotone valuations, there exists an allocation of $G$ in which each agent $i$ receives a connected piece with value at least $\emph{MMS}_i^{n+\fvsnum(G)}$.
\end{theorem}

\begin{proof}
Let $r := \fvsnum(G)$.
For each agent, consider her $1$-out-of-$(n+r)$ maximin partition.
Pick a subset of $r$ vertices upon whose deletion the remaining graph is a forest.
Delete each of these vertices (while keeping its adjacent edges intact as open intervals).
By the empty-intersection assumption, each vertex deletion harms at most one part in each agent's partition.
Therefore, once the graph becomes a forest, for every agent, at least $n$ parts remain.
By Theorem \ref{thm:forest-0}, there is an allocation in which every agent $i$ gets at least one of her maximin parts, and therefore value at least $\mms_i^{n+r}$.
\end{proof}
For every graph $G$, let $\mmsrank(G)$ be the smallest integer $r\geq 0$ such that
for any integer $n\geq 1$ and any $n$ agents with arbitrary monotone valuations, there exists an allocation of $G$ in which each agent $i$ receives a connected piece with value at least $\mms_i^{n+r}$.
Theorem~\ref{thm:graph-0} shows that 
$\mmsrank(G)\leq \fvsnum(G)$; we do not know if this inequality is tight.

When agents' valuations are \emph{additive},  Theorem~\ref{thm:graph-0} is \emph{not} tight for some graphs. In particular, 
when $G$ is a cycle, $\fvsnum(G)=1$, but it is known that  $\mms_i^{n}$ can be guaranteed to all agents (by reduction to an interval, for which proportionality can be guaranteed). 
Therefore, $\mmsrankadd(G) = 0$,
where $\mmsrankadd(G)$ is defined analogously to $\mmsrank(G)$ for additive valuations. 

\begin{open}
\label{open:mms-rank}
(a) Are there classes of graphs $G$ for which 
$\mmsrank(G) < \fvsnum(G)$?

(b) Are there graphs $G$ for which 
$\mmsrankadd(G) > 0$
(that is, for some $n$ agents with additive valuations, a $1$-out-of-$n$ maximin allocation does not exist)?
\end{open}

\section{Positive Separation}
\label{sec:positive-separation}

In this section, we consider the case where a separation constraint is imposed, that is, $s > 0$.

\subsection{Cutting Forests}
First, we assume that $G$ is a forest.
We start with several lemmas about the $\distg$ metric when $G$ is a tree.\footnote{
In this case, the metric space $\distg$ is known as a \emph{real tree} or an \emph{$\mathbb{R}$-tree} \citep{Bestvina01}.
}

By definition of a tree, for any two points $x,y\in G$, there is a unique (simple) path between $x$ and $y$.
Denote this unique path by $\pathg(x,y)$ or $x\to y$, and observe that the length of this path is $\distg(x,y)$.
We say that two subsets of $G$ are \emph{essentially-disjoint} if they intersect in at most a single point.

\begin{lemma}
\label{lem:unique-nearest}
Let $G$ be a tree, $X\subseteq G$  a closed connected subset, and $r\in G$ a point.
There exists a unique point $x_*\in X$ (a function of $X$ and $r$) satisfying the following properties:

(a) The path from any point in $X$ to $r$ passes through $x_*$. That is, for any $y\in X$: $x_* \in \pathg(y,r)$.

(b) $x_*$ is closer to $r$ than any other point in $X$ is. That is, $\distg(x_*,r) < \distg(y,r)$  for all $y\in X$.

(c) For any $y\in X$, $\distg(y,r) = \distg(y,x_*)+\distg(x_*,r)$.
\end{lemma}
\begin{proof}
If $r\in X$, then all three claims hold trivially by taking $x_*=r$. Assume therefore that $r\not\in X$.

(a)
For each point $y\in X$, denote by $y_*$ the unique point on $X\cap \pathg(y,r)$ that is closest to $r$ (intuitively, the point at which $\pathg(y,r)$ leaves $X$ and heads towards $r$).
We claim that this point is the same for all points in $X$, i.e., if $y,z \in X$ then $y_*=z_*$, as in the illustration below, where we denote this common point by $x_*$:

\begin{center}
\begin{tikzpicture}[scale=0.8]
\draw[line cap=round, gray!40, line width=5mm] (-3,0)--(5,0);
\node (r)  at (0,1.5) {$r$};
\node (xs) at (0,0) {$x_*$};
\node (y)  at (-3,0) {$y$};
\node (z)  at (5,0) {$z$};
\node (X)  at (1,-0.6) {$X$};
\draw (y) -- (xs) -- (r);
\draw (z) -- (xs) -- (r);
\end{tikzpicture}
\end{center}

Suppose for contradiction that $y_*\neq z_*$. 
The paths $\pathg(y_*,r)$ and $\pathg(z_*,r)$ meet at $r$. 
Let $w$ be the first point at which they meet. 
So the paths $\pathg(y_*,w)$ and $\pathg(z_*,w)$ are essentially-disjoint (they intersect only at $w$).
As $X$ is connected, there is a path $\pathg(y_*,z_*) \subseteq X$; this path is essentially-disjoint from both $\pathg(y_*,w)$ and $\pathg(z_*,w)$, since these two paths intersect $X$ only at their starting point $y_*$ and $z_*$, respectively.
Therefore, we have a cycle $y_* \to w \to z_* \to y_*$, as in the following figure:

\begin{center}
\begin{tikzpicture}[scale=0.8]
\draw[line cap=round, gray!40, line width=5mm] (-3,0)--(5,0);
\node (r) at (0,2) {$r$};
\node (w) at (0,1) {$w$};
\node (y) at (-3,0) {$y$};
\node (ystar) at (-1,0) {$y_*$};
\node (z) at (5,0) {$z$};
\node (zstar) at (2,0) {$z_*$};
\node (X) at (0.5,-0.6) {$X$};
\draw (y) -- (ystar) -- (w) -- (r);
\draw (z) -- (zstar) -- (w) -- (r);
\draw (zstar) -- (ystar);
\end{tikzpicture}
\end{center}

This contradicts the assumption that $G$ is a tree.
Parts (b) and (c) follow immediately from (a).
\end{proof}

Denote the unique point $x_*$ guaranteed by  Lemma \ref{lem:unique-nearest} by $\med(X,r)$.
\footnote{
$\med(X,r)$ is closely related to the concept of \emph{median} in a tree. Given three points $x, y, z$ of a tree graph, there is a unique point in $\pathg(x,y)\cap \pathg(y,z) \cap \pathg(z,x)$; this point is called the \emph{median} of $x,y,z$.
More generally, any graph with this uniqueness property is called a \emph{median graph}; such graphs have been studied in voting theory \citep{NehringPu07}.
}
The lemma can be generalized as follows:

\begin{lemma}
\label{lem:unique-nearest-subset}
Let $G$ be a tree and $X, R\subseteq G$ be closed connected subsets with $X\cap R=\emptyset$.
There exists a unique point $x_*\in X$ (a function of $X$ and $R$) satisfying the following properties:

(a) The path from any point in $X$ to any point in $R$ passes through $x_*$.

(b) For every point $r\in R$, $x_*$ is closer to $r$ than any other point in $X$ is.

\end{lemma}

\begin{proof}
As in Lemma~\ref{lem:unique-nearest}, it suffices to prove part (a).
For each point $y\in X$, denote by $r_y\in R$ the unique point such that any path from $y$ to a point in $R$ must pass through $r_y$; the existence of $r_y$ follows from Lemma~\ref{lem:unique-nearest}.
Denote by $y_*$ the unique point on $X\cap \pathg(y,r_y)$ that is closest to $R$ (i.e., the point at which $\pathg(y,r_y)$ leaves $X$ towards $R$).
We claim that this point is the same for all points in $X$, i.e., if $y,z \in X$ then $y_*=z_*$, as in the illustration below, where we denote this common point by $x_*$:

\begin{center}
\begin{tikzpicture}[scale=0.8]
\draw[line cap=round, gray!40, line width=5mm] (-3,0)--(5,0);
\draw[line cap=round, gray!40, line width=5mm] (-3,1.5)--(5,1.5);
\node (R) at (1,2.1) {$R$};
\node (rstar) at (0,1.5) {$r_*$};
\node (R0) at (-3,1.5) {};
\node (R1) at (5,1.5) {};
\node (xstar) at (0,0) {$x_*$};
\node (X0) at (-3,0) {};
\node (X1) at (5,0) {};
\node (X) at (1,-0.6) {$X$};
\draw (X0) -- (xstar) -- (rstar);
\draw (X1) -- (xstar) -- (rstar);
\draw (R0) -- (rstar) -- (R1);
\end{tikzpicture}
\end{center}

Suppose for contradiction that $y_*\neq z_*$. 
If the paths $\pathg(y,r_y)$ and $\pathg(z,r_z)$ intersect at some point (in particular, this happens if $r_y = r_z$), then there is a cycle in $G$ as in the proof of Lemma~\ref{lem:unique-nearest}.
Otherwise, the paths intersect $R$ at different points $r_y\ne r_z$, as in the following figure:

\begin{center}
\begin{tikzpicture}[scale=0.8]
\draw[line cap=round, gray!40, line width=5mm] (-3,0)--(5,0);
\draw[line cap=round, gray!40, line width=5mm] (-3,1.5)--(5,1.5);
\node (R) at (0.5,2.1) {$R$};
\node (ry) at (-1,1.5) {$r_y$};
\node (rz) at (2,1.5) {$r_z$};
\node (R0) at (-3,1.5) {};
\node (R1) at (5,1.5) {};

\node (y) at (-3,0) {$y$};
\node (ystar) at (-1,0) {$y_*$};
\node (z) at (5,0) {$z$};
\node (zstar) at (2,0) {$z_*$};
\node (X) at (0.5,-0.6) {$X$};
\draw (y) -- (ystar) -- (ry);
\draw (z) -- (zstar) -- (rz);
\draw (ystar) -- (zstar);
\draw (R0) -- (ry) -- (rz) -- (R1);
\end{tikzpicture}
\end{center}

Since $R$ and $X$ are connected and disjoint from each other, we have a cycle $y_* \to z_* \to r_z \to r_y \to y_*$, contradicting the assumption that $G$ is a tree.
\end{proof}

For non-intersecting subsets $X,R\subseteq G$, denote the unique point $x_*$ guaranteed by Lemma~\ref{lem:unique-nearest-subset} by $\med(X,R)$. 
Note that $\med(X,R)\neq\med(R,X)$: the former is in $X$ while the latter is in $R$.
We now define ``$s$-good'' pieces similarly to $0$-good pieces in Lemma~\ref{lem:tree-0}.
Given a graph $G$ and a family $\mathbf{X} := (X_1,\ldots,X_k)$ of closed connected pieces of $G$, a piece $X_{j^*}$ is called \emph{$s$-good} provided that
for all $j_1,j_2\in[k]$, the following holds: If $\distg(X_{j_1}, X_{j^*})<s$ and  $\distg(X_{j_2},X_{j^*})<s$, then
$\distg(X_{j_1},X_{j_2})<s$.

\begin{lemma}
\label{lem:tree-s}
Let $G$ be a tree and  $\mathbf{X} := (X_1,\ldots,X_k)$ a family of closed connected subsets of $G$, for some integer $k\geq 1$.
If $s>0$, then for some $j^*\in[k]$, the piece $X_{j^*}$ is $s$-good.
\end{lemma}
\begin{proof}
Fix an arbitrary point $r\in G$ as the tree \emph{root}. 
For every $j\in[k]$, let $x_j := \med(X_j,r)$ and $d_j := \distg(x_j,r)$.
Let $j^*\in \arg\max_{j\in[k]} d_j$, 
so that $X_{j^*}$ is a piece in $\mathbf{X}$ farthest from $r$; we abuse notation slightly and refer to this piece as $X_0$.
We claim that $X_0$ is $s$-good.
To prove this, we need several auxiliary claims on $X_0$.

\begin{claim}
\label{claim:intersect}
For each $j\in [k]$, if $X_0$ intersects $X_j$, then $x_0\in X_j$.
\end{claim}

\begin{proof}
Take any point $y\in X_0\cap X_j$. By Lemma \ref{lem:unique-nearest}, $\pathg(y,r)$ passes through both $x_0$ and $x_j$. By the selection of $x_0$, the path passes through $x_0$ before it passes through $x_j$, as in the illustration below:
\begin{center}
\begin{tikzpicture}[scale=0.8]
\draw[line cap=round, gray!40, line width=5mm] (1.2,1.4)--(3,1.9)--(4.5,1.9);
\draw[line cap=round, gray!40, line width=5mm] (1.2,2.6)--(3,2.1)--(6,2.1);
\node (r) at (7.5,2) {$r$};
\node (xj) at (6,2) {$x_j$};
\node (x0) at (4.5,2) {$x_0$};
\node (y) at (3,2) {$y$};
\node (Xj) at (1.2,2.5) {};
\node (Xjtext) at (3.5,2.7) {$X_j$};
\node (X0) at (1.2,1.5) {};
\node (X0text) at (3.5,1.2) {$X_0$};
\draw (X0) -- (y) -- (x0) -- (xj) -- (r);
\draw (Xj) -- (y);
\end{tikzpicture}
\end{center}
\vspace{-2mm}
It follows that $x_0\in \pathg(y, x_j) \subseteq X_j$. 
\end{proof}

\begin{claim}
\label{claim:not-intersect}
For each $j\in [k]$, if $X_0$ does not intersect $X_j$, then $x_0 =  \med(X_0,X_j)$.
\end{claim}

\begin{proof}
Let $y_0 := \med(X_0,X_j)$ and suppose for contradiction that $x_0\neq y_0$. Then there are two different paths from $x_0$ to $r$, as illustrated below,
where $y_j := \med(X_j,X_0)$:

\begin{center}
\begin{tikzpicture}[scale=0.8]
\draw[line cap=round, gray!40, line width=5mm] (-3,0)--(3,0);
\draw[line cap=round, gray!40, line width=5mm] (4,0)--(4,3);

\node (r) at (0,3) {$r$};

\node (x0) at (0,0) {$x_0$};
\node (X0L) at (-3,0) {};
\node (X0R) at (3,0) {};
\node (X0)  at (0,-0.6) {$X_0$};
\draw (x0) -- (r);

\node (xj) at (4,2) {$x_j$};
\node (XjL) at (4,0) {};
\node (XjR) at (4,3) {};
\node (Xj)  at (4.65,1.5) {$X_j$};
\draw (xj) -- (r);

\node (y0) at (2,0) {$y_0$};
\node (yj) at (4,1) {$y_j$};
\draw (x0) -- (y0) -- (yj) -- (xj) -- (r);

\draw (X0L) -- (x0) -- (y0) -- (X0R);
\draw (XjL) -- (yj) -- (xj) -- (XjR);

\end{tikzpicture}
\end{center}
The path $x_0 \to r$ intersects $X_0$ only at $x_0$.
The path $x_0\to y_0$ is contained in $X_0$;
$y_0\to y_j$ is essentially-disjoint from both $X_0$ and $X_j$; $y_j\to x_j$ is contained in $X_j$ (it may be empty); and $x_j\to r$ is essentially-disjoint from $X_j$. By the selection of $x_0$, $\distg(x_j,r)\leq \distg(x_0,r)$, so this last part $x_j\to r$ is essentially-disjoint from $X_0$ too. Therefore, these are two different paths from $x_0$ to $r$, contradicting that $G$ is a tree.
\end{proof}

For every $j\in[k]$ such that $X_0\cap X_j=\emptyset$, let $y_j := \med(X_j, X_0) =  \med(X_j, x_0)  $ and $z_j := \med( \pathg(x_0, y_j), r)$, i.e., $z_j$ is the point at which the path from $x_0$ to $r$ meets the path from $y_j$ to $r$. Our next auxiliary claim is (see also the figures in the proof):

\begin{claim}
\label{claim:distance-ineq}
For each $j\in[k]$, $\distg(y_j,z_j) \leq \distg(x_0,z_j)$.
\end{claim}

\begin{proof}
Consider first the case that $y_j = x_j$ (see the illustration below).
We have 
\[\distg(x_j,r) = \distg(x_j,z_j)+\distg(z_j,r)\]
and 
\[\distg(x_0,r)=\distg(x_0,z_j)+\distg(z_j,r).\]
Recall that $x_0$ was selected so that $\distg(x_0,r)\ge \distg(x_j,r)$. Subtracting the common term $\distg(z_j,r)$ gives $\distg(x_0,z_j) \ge \distg(y_j,z_j) = \distg(x_j,z_j)$. 

\begin{center}
\begin{tikzpicture}[scale=0.8]
\draw[line cap=round, gray!40, line width=5mm] (-6,2)--(-3,2);
\draw[line cap=round, gray!40, line width=5mm] (1.5,2)--(4,2);
\node (r) at (0,3.5) {$r$};
\node (zj) at (0,2) {$z_j$};
\node (x0) at (-3,2) {$x_0$};
\node (X0) at (-4.5,1.4) {$X_0$};
\node (xj) at (2,2) {$x_j=y_j$};
\node (Xj) at (3,1.4) {$X_j$};
\draw (-6,2) -- (x0) -- (zj)  -- (r);
\draw (4,2) -- (xj) -- (zj);
\end{tikzpicture}
\end{center}

Next, consider the case that $y_j\neq x_j$. 
Since $y_j=\med(X_j,x_0)$, it is contained in the path $x_0\to x_j$.
On the other hand, since $x_j=\med(X_j,r)$, it is contained in the path $y_j\to r$. Therefore, the path from $x_0$ to $r$ passes through both $y_j$ and $x_j$, as in the figure below:

\begin{center}
\begin{tikzpicture}[scale=0.8]
\draw[line cap=round, gray!40, line width=5mm] (-4,1)--(0,1);
\draw[line cap=round, gray!40, line width=5mm] (0,3)--(0,2)--(4,2);
\node (r) at (0,4) {$r$};
\node (xj) at (0,3) {$x_j$};
\node (yj) at (0,2) {$y_j$};
\node (x0) at (0,1) {$x_0$};
\node (Xj) at (2,2.6) {$X_j$};
\node (X0) at (-2,0.4) {$X_0$};
\draw (-4,1) -- (x0) -- (yj) -- (xj)  -- (r);
\draw (4,2) -- (yj);
\end{tikzpicture}
\end{center}

This implies that $z_j = \med(\pathg(x_0,y_j),r)=y_j$. Hence, $\distg(y_j,z_j) = 0 \leq \distg(x_0,z_j)$.
\end{proof}

By Claims~\ref{claim:intersect} and \ref{claim:not-intersect} and the definition of $y_j$, we get the following useful formula for the distance between $X_0$ and $X_j$, for every $j\in[k]$:
\begin{equation}
\label{eq:distance}
\distg(X_0,X_j) = \distg(x_0,X_j) = \distg(x_0,y_j),
\end{equation}
where the second equality holds if $X_0$ and $X_j$ are disjoint.
In other words, to measure the distance between the sets $X_0$ and $X_j$, we can consider the distance between the single point $x_0\in X_0$ (the point closest to $r$) and $X_j$.
If $X_0\cap X_j = \emptyset$, so that $y_j$ is defined,
then we can consider the distance between $x_0$ and the single point $y_j\in X_j$ (the point closest to $X_0$).

We are finally ready to prove that $X_0$ is $s$-good.
Consider two arbitrary pieces of $\mathbf{X}$, say $X_1$ and $X_2$, such that $\distg(X_1,X_0)<s$ and $\distg(X_2,X_0)<s$.
We have to prove that 
$\distg(X_1,X_2)<s$.

If $X_0 \cap X_1\neq \emptyset$, then $x_0\in X_1$ by Claim~\ref{claim:intersect}, so
\begin{align*}
\distg(X_1,X_2) &\leq \distg(x_0,X_2) \\
&= \distg(X_0,X_2)<s,
\end{align*}
where the equality holds by (\ref{eq:distance}).
An analogous claim holds if $X_0 \cap X_2\neq \emptyset$.
If $X_1\cap X_2\neq\emptyset$, then obviously $\distg(X_1,X_2) = 0 < s$.
So from now on suppose that $X_0 \cap X_1  = X_0\cap X_2 = X_1\cap X_2 = \emptyset$.
Then $y_1$ and $y_2$ are defined, and by (\ref{eq:distance}) we have $\distg(y_1,x_0)<s$ and $\distg(y_2,x_0)<s$.

By definition of $z_1$ and $z_2$, 
the path $x_0\to r$ must pass through both $z_1$ and $z_2$.
Without loss of generality, suppose that $z_1$ comes no later than $z_2$ on this path, so $\distg(x_0,z_1)\leq \distg(x_0,z_2)$, as in the illustration below:
\begin{center}
\begin{tikzpicture}[scale=0.8]
\node (x0) at (0,0) {$x_0$};
\node (x1) at (1,-1) {$y_1$};
\node (z1) at (2,0) {$z_1$};
\node (x2) at (3,1) {$y_2$};
\node (z2) at (4,0) {$z_2$};
\node (r) at (6,0) {$r$};
\draw (x0) -- (z1) -- (z2) -- (r);
\draw (x1) -- (z1);
\draw (x2) -- (z2);
\end{tikzpicture}
\end{center}
Consider the path $y_1 \to z_1 \to z_2 \to y_2$. The length of this path is at most
\begin{align*}
&\distg(y_1,z_1)+
\distg(z_1,z_2)+
\distg(z_2,y_2)
\\
&\leq 
\distg(x_0,z_1)+
\distg(z_1,z_2)+
\distg(z_2,y_2)
\\
&= 
\distg(x_0,y_2) \\
&= 
\distg(X_0,X_2)
< s,
\end{align*}
where the first inequality holds by Claim~\ref{claim:distance-ineq} and the last equality by (\ref{eq:distance}).

We have demonstrated a path of length shorter than $s$ from a point $y_1\in X_1$ to a point $y_2\in X_2$;
this proves that $\distg(X_1,X_2)<s$.
It follows that $X_0$ is $s$-good.
\end{proof}

With Lemma~\ref{lem:tree-s} in hand, we can now establish the existence of a maximin allocation for forests by using similar arguments as in Theorem~\ref{thm:forest-0}.
In particular, when we allocate an $s$-good part $T_{i,j}$ to agent $i$, we remove all portions of the tree that are within distance $s$ of $T_{i,j}$, and divide the remaining cake recursively among the remaining agents.

\begin{theorem}
\label{thm:forest-s}
Let $G$ be a forest and $s > 0$.
For agents with arbitrary monotone valuations, 
a maximin allocation exists.
\end{theorem}

\subsection{Cutting General Graphs}

We now proceed to general graphs.
Even when the graph is a simple cycle, \citet{ElkindSeSu21} showed that a maximin allocation does not necessarily exist, but the $1$-out-of-$(n+1)$ maximin share can be guaranteed.
We present here a more general theorem analogous to Theorem~\ref{thm:graph-0}.
However, our theorem requires the assumption that the length of each edge is at least $s$.
We remark that the separation parameter $s$ is generally small in our motivating applications such as transition or buffer zones, so this assumption is realistic.
Nevertheless, it is an interesting question whether the result continues to hold without this assumption.

\begin{theorem}
\label{thm:graph-s}
Let $s > 0$.
Let $G$ be a graph in which the length of each edge is at least $s$. 
For any $n$ agents with arbitrary monotone valuations,
there exists an $s$-separated allocation in which each agent $i$ receives a connected piece with value at least $\emph{MMS}_i^{n+\fvsnum(G)}$.
\end{theorem}

\begin{proof}
Let $r := \fvsnum(G)$.
Let $u_1,\ldots, u_r$ be a set of vertices such that, when they are deleted from $G$, the remaining graph is a forest.
For each $j\in[r]$, remove an open interval of length $s/2$ from each edge adjacent to $u_j$.
Consider the $1$-out-of-$(n+r)$ maximin partition of each agent.
For each $j\in[r]$, the distance between each pair of points in the set removed due to the vertex $u_j$
is less than $s$. 
Hence, the removed set overlaps at most one part of each agent's partition. 
Therefore, once the graph becomes a forest, for every agent, at least $n$ parts remain.
By Theorem~\ref{thm:forest-s}, there exists an $s$-separated allocation in which every agent $i$ gets at least one of her maximin parts, and therefore value at least $\mms_i^{n+r}$.
Then, reconstruct the original graph by putting back the removed intervals of length $s/2$.
The allocation remains $s$-separated.
\end{proof}

As in the case $s=0$ (see Open Problem \ref{open:mms-rank}), we do not know whether the factor $n+\fvsnum(G)$ is tight in general.
Below, we present a class of graphs for which we can obtain a tight bound.
In particular, we consider the family of graphs 
such that the feedback vertex set number of each connected component is at most $1$ (that is, every connected component is either a tree, or can be made acyclic by removing a single vertex).
We denote $N := \min(n+\fvsnum(G),~ 2n - 1)$.

\begin{theorem}
\label{thm:forest-cycles-s}
The following statements hold for any real number $s > 0$ and integer $n\geq 1$: 

(a) Let $G$ be a vertex-disjoint union of graphs with FVS number $\leq 1$, such that the length of each edge is at least $s$. 
For any $n$ agents with monotone valuations, there is an allocation in which each agent $i$ receives value at least $\emph{MMS}_i^{N}$.

(b) For any integer $r\geq 1$, there exists a graph $G$ with $\fvsnum(G) = r$ (specifically, a union of $r$ cycles and zero or more trees), 
and $n$ agents with additive valuations, such that no allocation gives each agent $i$ at least $\emph{MMS}_i^{N-1}$.
\end{theorem}

\begin{proof}
We prove the two parts in turn.
For part (a), consider a $1$-out-of-$N$ maximin partition of each agent.
Define a bipartite graph $H = (X, Y, E)$, where $X$ is the set of agents, $Y$ is the set of connected components of $G$ with FVS number $1$, and a component $C$ is adjacent to agent~$i$ if and only if $C$ contains at least one part from $i$'s maximin partition. 

Every bipartite graph admits unique partitions $X = X_S\cup X_L$ and $Y = Y_S\cup Y_L$ such that the following holds
(see, e.g., Theorem 1.3 of \citet{aigner2019envy}):
\begin{center}
\includegraphics[width=6cm]{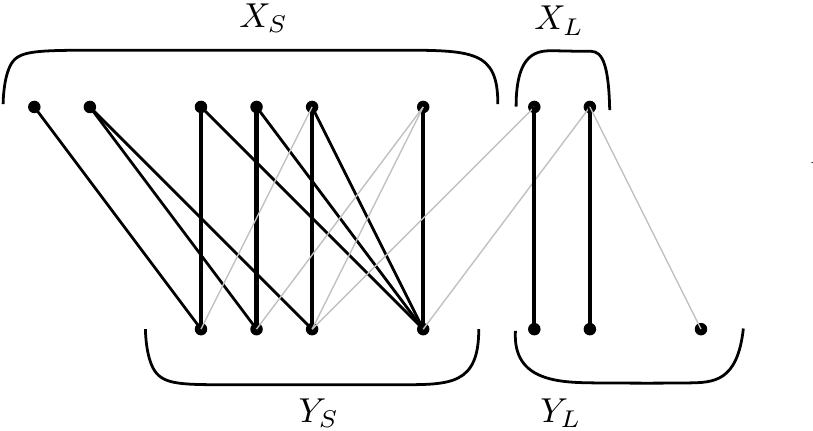}
\end{center}
\begin{enumerate}[(i)]
\item There are no edges between $X_S$ and $Y_L$;
\item The subgraph $G[X_S,Y_S]$ is ``$Y$-path-saturated'', which, for our purposes, just implies that $|X_S| > |Y_S|$;
\item The subgraph $G[X_L,Y_L]$ admits a matching that saturates all vertices of $X_L$.
\end{enumerate}
With this partition at hand, the cake is allocated as follows.

Every agent in $X_L$ receives an entire component from $Y_L$ according to the matching in (iii). 
It follows that every such agent receives a value of at least $\mms_i^{N}$.

Let $n' := |X_S|$ and $G'$ be the subgraph of $G$ containing the components in $Y_S$ along with all tree components of $G$. 
By (i), all parts in the maximin partition of every agent in $X_S$ are contained in $G'$.
By Theorem \ref{thm:graph-s}, 
$G'$ admits an allocation in which each agent $i\in X_S$ receives value at least $\mms_i^{n' + \fvsnum(G')}$.
Clearly, $n' + \fvsnum(G') \leq n + \fvsnum(G)$.
Since the FVS number of
each element of $Y_S$ equals $1$, we have
$\fvsnum(G') = |Y_S| \le n'-1$ by (ii), so 
$n' + \fvsnum(G') \leq 2 n' -1\leq 2 n -1$. Therefore, 
$n' + \fvsnum(G') \leq N$, 
so $\mms_i^{n' + \fvsnum(G')} \geq \mms_i^{N}$.
This completes the proof of part (a).

We now proceed to part (b).
Given integers $n\geq 1$ and $r\geq 1$, 
we denote $N = \min(n+r, 2n - 1)$.
We construct a graph $G$ made of a disjoint union of $r$ cycles, and valuation functions of $n$ agents such that for all agents~$i$, $\mms_i^{N-1}=1$, but every $s$-separated allocation gives a positive value to at most $n-1$ agents.

\underline{Case 1:} $r\ge n$.
In this case, we have $N-1 = 2n-2$.
All $r$ cycles are of length $2 s + 2\varepsilon$ for $0 < \varepsilon \ll s$. Some $n-1$ cycles are ``valuable'', i.e., each agent $i$ values every such cycle at $2$, with a valuation evenly concentrated in two regions of length $\varepsilon$ each: one at angle $i\pi/n$ and one at angle $i\pi /n + \pi$ (radians).
The other $r-(n-1)$ cycles have no value to any agent.
Every agent can partition every valuable cycle into two $s$-separated regions of value $1$, so $\mms_i^{2n-2} = 1$ for every agent $i\in [n]$.
But from every valuable cycle, a positive value can be allocated to at most one agent, so all in all, at most $n-1$ agents can get a positive value. 

\underline{Case 2:} $r\leq n-1$.
In this case, we have $N-1 = n+r-1$.
Some $r-1$ cycles are ``small'', with length $2 s + 2\varepsilon$ for $0< \varepsilon\ll s$, and the agents' valuations as in Case~1.
The $r$-th cycle is ``large'', with length
$(n+1-r)\cdot (s+\varepsilon)$.
Each agent $i$ values the large cycle at $n+1-r$, with a valuation concentrated in $n+1-r$ small $s$-separated regions of length $\varepsilon$ each. 
Therefore, every agent can partition the large cycle into $n+1-r$ regions of value $1$ which are $s$-separated. 
By adding two regions for each of the $r-1$ small cycles, we get $\mms_i^{n+r-1} = 1$ for all $i\in [n]$.

The valuable regions on the large cycle are arranged in such a way that at most $n-r$ agents can receive a positive value from the large cycle. 
In particular, for each $i\in[n-1]$, the regions of agent $i+1$ are shifted clockwise from those of agent~$i$ by a small amount, say $\varepsilon$.
This is illustrated in the figure below, where $r=1$, $n=4$, $n+1-r=4$, and the length of each side of the square is $s+\varepsilon$. 
Each color denotes the valuable regions of a single agent.

\begin{center}
\begin{tikzpicture}[scale=3]
\draw (0,0) -- (0,1.1) -- (1.1,1.1) -- (1.1,0) -- (0,0);  
\draw[line width=2mm, green] (0,0.1)--(0,0.167);
\draw[line width=2mm, blue] (0,0.2)--(0,0.267);
\draw[line width=2mm, red] (0,0.3)--(0,0.367);
\draw[line width=2mm, purple] (0,0.4)--(0,0.467);

\draw[line width=2mm, green] (0.1,1.1)--(0.167,1.1);
\draw[line width=2mm, blue] (0.2,1.1)--(0.267,1.1);
\draw[line width=2mm, red] (0.3,1.1)--(0.367,1.1);
\draw[line width=2mm, purple] (0.4,1.1)--(0.467,1.1);

\draw[line width=2mm, green] (1.1,1.0)--(1.1,0.933);
\draw[line width=2mm, blue] (1.1,0.9)--(1.1,0.833);
\draw[line width=2mm, red] (1.1,0.8)--(1.1,0.733);
\draw[line width=2mm, purple] (1.1,0.7)--(1.1,0.633);

\draw[line width=2mm, green] (1.0,0)--(0.933,0);
\draw[line width=2mm, blue] (0.9,0)--(0.833,0);
\draw[line width=2mm, red] (0.8,0)--(0.733,0);
\draw[line width=2mm, purple] (0.7,0)--(0.633,0);
\end{tikzpicture}
\end{center}

As at most one agent can receive a positive value from each of the $r-1$ small cycles, at most $(r-1)+(n-r) = n-1$ agents overall can receive a positive value.
This completes the proof of Case~2, and therefore the proof of part (b).

Note that the impossibility result can be extended to graphs with any number of trees (in addition to the $r$ cycles), by simply assuming that all trees have a value of 0 to all agents.
\end{proof}

Theorem~\ref{thm:forest-cycles-s} shows that Theorem~\ref{thm:graph-s} is tight for unions of trees and cycles when the number of cycles is less than $n$.
However, the exact factor for other graphs remains open.

\section{Discussion}
\label{sec:conclusion}

In this work, we have studied the division of a graphical cake using the maximin share notion, both with and without separation constraints.
Our most technically challenging result shows that a maximin allocation exists for positive separation whenever the graph is acyclic.
A tempting approach to simplify this proof is by using the ``last diminisher'' method, wherein each agent can trim a proposed piece as long as the remaining piece after trimming still yields value at least the agent's maximin share.
Indeed, this method was used by \citet{BouveretCeEl17} to establish the existence result for trees in the context of indivisible goods, and the algorithm of \citet{ElkindSeSu21} for interval cakes can also be seen as a version of last diminisher.
We remark here that the approach does \emph{not} work in our setting with separation.
Indeed, consider the subtree in Figure~\ref{fig:bad-tree}, where two of Alice's parts in her maximin partition are bold, while one of Bob's parts is dashed.
The last diminisher method may allocate Bob's part to him; however, when we take the separation requirement into account, we cannot allocate either of Alice's parts in its entirety. 
Note that the dashed piece is not $s$-good.
This example precisely demonstrates why, in the proof of Theorem~\ref{thm:forest-s}, we need the elaborate lemmas on real trees in order to guarantee the existence of an $s$-good piece.

\vspace{-0.5mm}
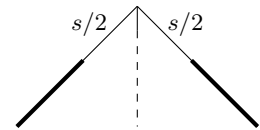
\begin{figure}[!ht]
\centering
\begin{tikzpicture}[scale=0.8]
\draw [dashed] (3,2.6) -- (3,1);
\draw (3,3) -- (3,2.6);
\draw (3,3) -- (5,1);
\draw (3,3) -- (1,1);
\draw [ultra thick] (1,1) -- (2.1,2.1);
\draw [ultra thick] (5,1) -- (3.9,2.1);
\node at (2.2,2.7) {\small $s/2$};
\node at (3.8,2.7) {\small $s/2$};
\end{tikzpicture}
\caption{An example subtree in which the last diminisher method fails to produce a maximin allocation.}
\label{fig:bad-tree}
\end{figure}
\vspace{-0.5mm}

More generally, our work builds upon an active line of research that incorporates realistic constraints in fair division problems.
We believe that identifying and studying such considerations will lead to technically intriguing questions as well as practically useful fairness guarantees.

\section*{Acknowledgments}

This work was partially supported by the European Research Council (ERC) under grant number 639945 (ACCORD), by the Israel Science Foundation under grant number 712/20, and by an NUS Start-up Grant.
We would like to thank the anonymous reviewers for their valuable comments.

\bibliographystyle{named}
\bibliography{main}

\end{document}